\documentclass[12pt]{article}
\usepackage{amsmath, amssymb,amsthm,amscd,graphicx,dsfont}
\usepackage{hyperref}
\usepackage{subfigure}
\usepackage{bbding,phaistos,wasysym}
\usepackage[mathscr]{eucal}

\newcommand{\address}[1]{\vbox{\center\em#1}}

\usepackage{tikz}
\usetikzlibrary{arrows,decorations.pathmorphing,decorations.pathreplacing,decorations.markings,shapes.geometric,calc,positioning,chains,matrix,scopes} 

\numberwithin{equation}{section}
\theoremstyle{plain} 
\theoremstyle{definition}
\newtheorem{thm}{Theorem}[section]
\newtheorem{lem}[thm]{Lemma}

\theoremstyle{definition}
\newtheorem{defn}{Definition}[section]

\theoremstyle{remark}

\newcommand{\R}{\mathbb{R}}
\newcommand{\C}{\mathbb{C}}

\renewcommand{\P}{\mathbb{P}}
\newcommand{\Z}{\mathbb{Z}}

\DeclareMathOperator{\Tr}{Tr}
\DeclareMathOperator{\Spec}{Spec}

\DeclareMathOperator{\End}{\mathrm{End}}
\DeclareMathOperator{\Hom}{\mathrm{Hom}}
\DeclareMathOperator{\Irr}{\mathrm{Irr}}
\DeclareMathOperator{\rad}{\mathrm{rad}}

\DeclareMathOperator{\wt}{\mathrm{wt}}

\begin{document}
\begin{titlepage}

\begin{center}
\vskip 9mm
{\LARGE
Brane Tilings and Non-Commutative Geometry
}
\vskip 10mm
Richard Eager${}$\footnote{\href{mailto:reager@physics.ucsb.edu}{\rm
reager@physics.ucsb.edu}}
\vskip 5mm
\address{
${}$
Department of Physics, University of California,\\
Santa Barbara, CA 93106, USA
}
\vskip 6mm

\end{center}

\abstract{
\noindent \normalsize{
We derive the quiver gauge theory on the world-volume of D3-branes transverse to an $L^{a,b,c}$ singularity by computing the endomorphism algebra of a tilting object first constructed by Van den Bergh.   The quiver gauge theory can be concisely specified by an embedding of a graph into a face-centered cubic lattice.  In this description, planar Seiberg dualities of the gauge theory act by changing the graph embedding.  We use this description of Seiberg duality to show these quiver gauge theories possess periodic Seiberg dualities whose existence were expected from the AdS/CFT correspondence.
 } }

\vfill

\end{titlepage}

\setcounter{footnote}{0}

{\addtolength{\parskip}{-1ex}\tableofcontents}

\section{Introduction}
D-branes at singularities provide a geometric description of gauge theories and are a useful tool for string theory model building.  The low energy effective physics on a stack of D3-branes transverse to a Calabi-Yau singularity has a limit where the closed string modes decouple leaving only an effective gauge theory living on the D3-brane world-volume governed by open string modes.  Given a Calabi-Yau singularity, determining the explicit form of the Lagrangian describing the D3-brane world-volume gauge theory is important for model building applications and precise checks of gauge-gravity duality.  For a general Calabi-Yau manifold, the world-volume gauge theory will have $\mathcal{N}=1$ supersymmetry and can be encoded in a superpotential and K\"{a}hler potential.  While the K\"{a}hler potential depends on how the singularity is embedded inside a compact Calabi-Yau manifold, the superpotential depends only on the local complex structure of the singularity.  Although there are several methods to compute superpotentials, the method of Aspinwall and Katz \cite{Aspinwall:2004bs} is the only one applicable to arbitrary Calabi-Yau singularities.  Despite their method's great generality, its use has been limited by computational difficulties.  In practice, the method of brane tilings \cite{Franco:2005rj,Hanany:2005ve} has been used to determine the gauge theories for local toric Calabi-Yau singularities.  While only applicable to toric singularities, brane tilings have yielded several fascinating results whose relation to the Aspinwall-Katz method has so far been obscure.  In particular, the world-volume gauge theory of a toric singularity should be described by a periodic quiver.  For the infinite families of Calabi-Yau singularities given by cones over the Sasaki-Einstein spaces $Y^{p,q}$ and $L^{a,b,c}$, we will see how Aspinwall's method of constructing a tilting object \cite{Aspinwall:2008jk} leads to a periodic quiver.
For toric Calabi-Yau singularities that are cones over Fano surfaces, the relationship between brane tilings and tilting objects has recently been explained in \cite{Carqueville:2009xu}.

After Maldacena proposed the AdS/CFT correspondence between the world-volume gauge theory on coincident D3-branes in Minkowski space and type IIB string theory on $AdS_5 \times S^5$ \cite{Maldacena:1997re}, it was subsequently generalized to a correspondence between the world-volume gauge theory on coincident D3-branes transverse to a Calabi-Yau singularity and type IIB string theory on $AdS_5 \times L^5$, where $L^5$ is a Sasaki-Einstein manifold.  In particular, $L^5$ is found by taking the near-horizon limit on the stack of D3-branes and the original Calabi-Yau manifold is a cone with base $L^5.$  Shortly thereafter, the correspondence was extended to orbifolds of $\C^3$ \cite{Kachru:1998ys, Douglas:1996sw}.  Klebanov and Witten \cite{Klebanov:1998hh} showed how to determine the superpotential for the conifold singularity by matching the R-symmetries of the gauge theory to its AdS/CFT dual.  Morrison and Plesser \cite{Morrison:1998cs} re-derived the conifold superpotential starting with the superpotential of an orbifold singularity and following how it changed under partial resolutions of the singularity.  Their derivation was turned into a systematic method for computing the superpotential of any local toric manifold \cite{Beasley:1999uz} \cite{Feng:2000mi}.  Surprisingly, the quivers obtained through this algorithmic procedure could be simply described by dimer models and brane tilings \cite{Hanany:2005ve, Franco:2005rj}.  Brane tilings for several examples of $L^{a,b,c}$ quivers greatly simplified further checks of the AdS/CFT correspondence with the newly discovered metrics for the $Y^{p,q}$ \cite{Gauntlett:2004yd} and $L^{a,b,c}$ \cite{Cvetic:2005ft}  families of singularities \cite{Franco:2005sm,Benvenuti:2005ja,Butti:2005sw}.  A summary of more recent results on brane tilings is contained in the two excellent reviews \cite{Kennaway:2007tq} and \cite{Yamazaki:2008bt}.

Hanany, Herzog, and Vegh \cite{Hanany:2006nm} explain how to construct a brane tiling from an exceptional collection of line bundles.  Our construction is similar in spirit, except we work directly with the singular geometry instead of a smooth resolution.  The primary difficulty in either method is finding an exceptional collection or a tilting object.  Methods for computing tilting objects primarily rely on local cohomology and are described in \cite{MR1087057} and \cite{Aspinwall:2008jk}.  Tilting objects for the $L^{a,b,c}$ family of singularities were constructed by Van den Bergh  \cite{MR2057015, MR1087057}.  We will compute the endomorphism algebra of these titling objects and show how the endomorphism algebra determines a quiver gauge theory.

Extrapolating from several examples of $L^{a,b,c}$ gauge theories constructed in \cite{Franco:2005sm}, several of the properties of the $L^{a,b,c}$ quiver gauge theories have been determined.  We will verify that the total number of fields with a fixed R-charge is in agreement with the results in \cite{Franco:2005sm, Benvenuti:2005ja, Butti:2005vn}.  Furthermore, we will show that all of these quiver gauge theories possess Seiberg dualities that leave the quiver invariant after relabeling its nodes.  This provides a simple description of the quiver and is the first step in verifying the existence of a duality cascade whose existence was conjectured in \cite{Butti:2005sw, Brini:2006ej, Evslin:2007au}.  We expect that the cascade will have similar qualitative features to the cascade for the $Y^{p,q}$ family  \cite{Herzog:2004tr, Benvenuti:2004wx} and its supergravity dual \cite{Herzog:2004tr}.  

We first define periodic quivers and explain how they encode the structure of a quiver gauge theory.  In section \ref{sec:torus}, we will introduce the $L^{a,b,c}$ family of local toric geometries that we are studying.  To determine the associated quiver gauge theory, we introduce the method of tilting in section \ref{sec:tilting} and illustrate the method with the conifold and $Y^{2,1}$ in section \ref{sec:examples}.  In section \ref{sec:morphisms}, we show that the quiver obtained from tilting is periodic.  We find that all of the $L^{a,b,c}$ gauge theories can be described by defining a doubly periodic integer-valued height function on $\Z^2.$
In section \ref{sec:Seiberg} we show all planar Seiberg dualities can be described by increasing or decreasing the height function of a vertex by 2.  This allows us to show that all of the $L^{a,b,c}$ quiver gauge theories possess periodic Seiberg dualities whose existence is necessary for duality cascades to exist.  Finally, we suggest possible future extensions of this work in section \ref{sec:conclusion}.
\section{Quiver Gauge Theories}
The world-volume gauge theory on a stack of D3-branes at a Calabi-Yau singularity is often described by a quiver gauge theory.  A {\it quiver} $Q = (V,A,h,t: A \rightarrow V)$ is a collection of vertices $V$ and arrows $A$ between the vertices of the quiver.  The arrows are directed edges with the head and tail of an arrow $a \in A$ given by maps $h(a)$ and $t(a)$, respectively.  A {\it representation} $X$ of a quiver is an assignment of $\C$-vector spaces $X_{v}$ to every vertex $v \in V$ and a $\C-$linear map $\phi_{a}: X_{t(a)} \rightarrow X_{h(a)}$ to every arrow $a \in A.$  The {\it dimension vector} $\mathbf{n} \in \mathbb{N}^{|V|}$ of a representation $X$ is a vector with an entry for each vertex $v \in V$ equal to the dimension of the vector space $X_v.$

A {\it quiver gauge theory} is specified by a quiver and superpotential in the following manner:
\begin{itemize}
\item The gauge group
$$G = \prod_{v \in V} U(n_v)$$
is a product of unitary groups $U(n_v)$ of dimension $n_v.$
\item Arrows $a \in A$ represent chiral superfields $\Phi_a$ transforming in the fundamental representation of $U(n_{h(a)})$ and in the anti-fundamental representation of $ U(n_{t(a)})$.  If the two vertices are distinct the chiral superfields are called {\it bifundamental} fields.  Otherwise, the arrow is a loop and the field transforms in the adjoint representation.
\item The superpotential
$$W = \sum_{l = a_1 a_2 \dots a_k \in L} \lambda_{l} \Tr \left[ \Phi_{a_1} \Phi_{a_2} \dots \Phi_{a_k} \right]$$
is a sum of gauge invariant operators $\Tr \left[ \Phi_{a_1} \Phi_{a_2} \dots \Phi_{a_k} \right].$  Gauge invariance requires $l = a_1 a_2 \dots a_k$ to be an oriented loop in the quiver.  Each operator has coupling constant $\lambda_{l}.$
\end{itemize}
\begin{figure}
\begin{center}
\includegraphics[trim = 0mm 1mm 0mm 1mm, clip,width=12cm]{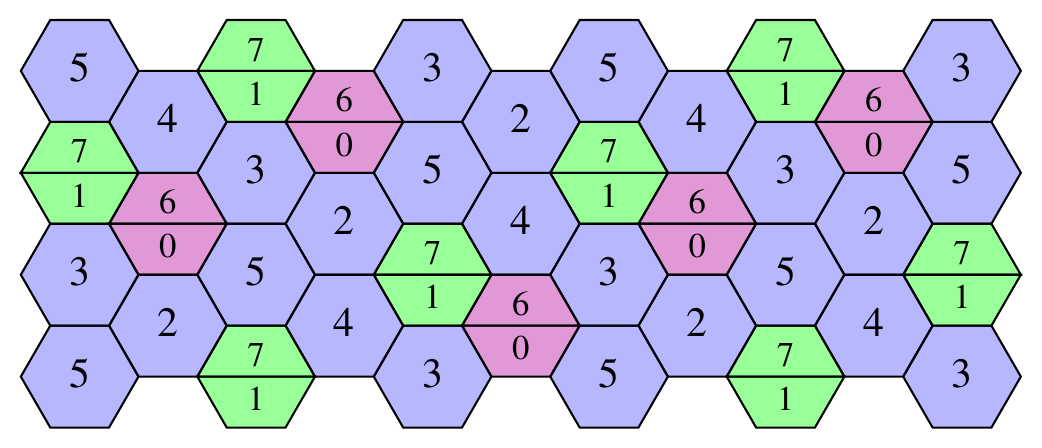}
\caption{Brane tiling for $L^{2,6,3}$}
\label{fig:263tiling}
\end{center}
\end{figure}

Brane tilings form a particularly simple class of quiver gauge theories whose superpotential is easily determined from the following graphical representation.  A {\it brane tiling} is a bipartite graph $G = (G_0^{\pm}, G_1)$ with an embedding into the two-torus such that the faces form a tiling of the torus.  An example is shown in figure \ref{fig:263tiling}.  From a brane tiling, we can form a {\it periodic quiver} $Q = (Q_0, Q_1,Q_2, h, t),$ which is the dual graph.  The vertices $Q_0$ of the periodic quiver are the faces of the brane tiling.  The faces $Q_2 = Q_2^{+} \cup Q_2^{-}$ are dual to the vertices of the brane tiling and the faces $Q_2^{+}$ and $Q_2^{-}$ are oriented clockwise and counterclockwise respectively.  The arrows $Q_1$ follow the orientation of the faces.  The functions $h: Q^1 \rightarrow Q^{0}$ and $t: Q^1 \rightarrow Q^{0}$ specify the head and tail of each arrow in the quiver.

Associated to any quiver is its path algebra $\C Q$, where the multiplication of two paths $\gamma_1$ and $\gamma_2$ is their concatenation $\gamma_2 \gamma_1$ if $t(\gamma_1) = h(\gamma_2)$ and zero otherwise.  The commutator subalgebra $[\C Q,\C Q]$ is spanned by cyclic words.  A superpotential $W$ is an element of $\C Q / [ \C Q, \C Q].$
A word in $\C Q / [ \C Q, \C Q]$ can be embedded into $\C Q$ by summing over all of its cyclic permutations and the map can be extended to the entire vector space by linearity.  The quiver gauge theory associated to a brane tiling has superpotential
$$W = \sum_{F \in Q_{2}^{+}} w_F - \sum_{F \in Q_2^{-}} w_F$$
where $Q_2^{+}$ and $Q_2^{-}$ are the positively and negatively oriented faces and $w_F$ is the product of fields going around a face $F.$
The F-term relations are encoded by the Jacobian ideal $( \partial W)$ and representations of the {\it superpotential algebra} $\mathcal{A} = \C Q/ (\partial W)$ encode the moduli space of vacua of a 4D $\mathcal{N} = 1$ supersymmetric field theory \cite{Luty:1995sd}.

For a quiver gauge theory to be physically sensible, the gauge anomalies for each gauge group must vanish.
Vanishing of the triangle anomaly with three external gluons of the $U(n_v)$ gauge group yields the condition
\begin{equation}
\sum_{a \in A | h(a) = v} n_{t(a)} - \sum_{a \in A | t(a) = v} n_{h(a)} = 0.
\label{eqn:rank}
\end{equation}
Linear combinations $U(1)_q$ of the $U(1)_v \subset U(n_v)$ groups can mix and lead to triangle anomalies of the form $\Tr \left[ SU(n_v)^2 U(1)_q \right].$  Vanishing of this mixed anomaly requires
\begin{equation}
\label{eqn:baryon}
\sum_{a \in A | h(a) = v} n_{t(a)} q_{t(a)} - \sum_{a \in A | t(a) = v} n_{h(a)} q_{n(a)} = 0.
\end{equation}
Quiver gauge theories describing the low energy effective field theory of D-branes at a Calabi-Yau singularity have a variant of the Green-Schwarz mechanism to cancel the anomalous $U(1)$'s.
The gauge fields of the anomalous $U(1)$'s couple to RR-form fields giving them St\"{u}ckelberg masses \cite{Douglas:1996sw, Ibanez:1998qp, Antoniadis:2002cs}.  These massive vector fields decouple in the IR.  The non-anomalous $U(1)$ fields are IR free so they also decouple and become global $U(1)$ symmetries in the IR.  This is explained from a large-volume perspective in \cite{Jockers:2004yj, Buican:2006sn, Martelli:2008cm}.
\section{Non-commutative Resolutions of Singularities}
\label{sec:torus}
D-branes act as remarkable probes of geometry.  Berenstein and Leigh \cite{Berenstein:2001jr} proposed that D-branes could be used to construct non-commutative resolutions of singularities.  A giant step toward the realization of their proposal was Van den Bergh's \cite{MR2077594} algebraic definition of a non-commutative crepant resolution.
\begin{defn}[\cite{MR2077594} section 4.1]
A non-commutative crepant resolution of a normal Gorenstein domain $R$ is a homologically homogeneous $R$-algebra of the form
$$A = \End_R(M)$$ where $M$ is a reflexive $R-$module.
\end{defn}
In the next section we will see how non-commutative resolutions of the $L^{a,b,c}$ singularities constructed by Van den Bergh can be used to determine a quiver gauge theory.  The $L^{a,b,c}$ singularities can be described by a geometric invariant theory quotient of the form $(\C^{4} - Z)/ \C^{*},$ where $Z$ is a set of points that must be removed to form a good quotient.  These spaces can equivalently be characterized as the moduli space of vacua in Witten's GLSM construction \cite{Witten:1993yc}.  We specify the $\C^{*}$ action on the ring $S = \C[\alpha_1, \dots , \alpha_m, \beta_1, \dots \beta_n]$ by
\begin{align*}
z \cdot \alpha_i & = z^{a_i} \alpha_i,  \text{ for } a_i \in \Z^{+} \\
z \cdot \beta_i & = z^{b_i} \beta_i,  \text{ for } b_i \in \Z^{-}
\end{align*}
where $z$ is the $\C^{*}$ coordinate.  To form a good geometrical invariant theory quotient we must excise either the set $Z = \left\{\alpha_i = 0 \right\}$ or $Z = \left\{\beta_i = 0 \right\}.$  Call the ring of invariants of $S$ under the $\C^{*} = T$ action $R = S^T.$  All elements of the ring $S$ with total weight $m$ form an $R-$module, denoted $S(m).$  Cox \cite{MR1299003} showed that these modules can be used to construct sheaves on the resolved geometry $(\C^{4} - Z)/ \C^{*}.$  We will compute directly with the modules, instead of their corresponding sheaves, to simplify the computation of the endomorphism algebra.  Write the sum of the positive and negative weights as $N^{+} = \sum_{i} a_i$ and $N^{-} = - \sum_{i} b_i$ respectively.  When $N^{+} = N^{-}$, the quotient $(\C^{n} -Z)/\C^{*}$ will have a Calabi-Yau resolution.  We restrict attention to this case and write $N$ for the common value of $N^{+}$ and $N^{-}$.
Van den Bergh's final result of \cite{MR2077594}  is:
\begin{thm}
\label{thm:vdb}
If $\sum_{i} a_i = - \sum_{i} b_i \equiv N,$ there are at least two positive and two negative weights, and the greatest common divisor of all of the weights is one, then the ring of invariants $R = S^T$ is Gorenstein and has a non-commutative crepant resolution
$$A = \End_R\left( \oplus_{m = 0}^{N-1} S(m) \right).$$
\end{thm}
Van den Bergh's theorem's relevance for constructing quiver gauge theories was first observed in \cite{Aspinwall:2008jk, Herbst:2008jq}.
The cones over the $L^{a,b,c}$ spaces are simply the quotients with weights $(a_1, a_2, b_1, b_2) =  (a,b,-c,-d)$ where $d = a+b-c$ \cite{Martelli:2005wy}.  The cones over the $Y^{p,q}$ family of singularities are the quotients with weight vectors $(p-q,p+q,-p,-p).$  We will see in the next section how the endomorphism algebra determines a quiver gauge theory for the space $(\C^{4}- Z)/(\C^{*}).$
\section{Quivers from Tilting}
\label{sec:tilting}
Van den Bergh's non-commutative resolution has a simple interpretation as a quiver with relations.  The endomorphism algebra can be described by a quiver with vertices labeled by the modules $S(m), m = 0, \dots N-1$ and arrows labeling irreducible morphisms between the vertices.  We now determine the irreducible morphisms in the endomorphism algebra following Keller's review \cite{keller08}.  The sum $$\mathcal{S} = \bigoplus_{m = 0}^{N-1} S(m)$$ is an example of a ``tilting complex''.
One of the key features of a tiling complex is that the higher Ext groups between its summands vanish.  The full definition of a tilting complexes is given in appendix \ref{app-Morita}.

The irreducible morphisms in $\End_R(\mathcal{S})$ are defined to be
$$\Irr_{\mathcal{S}}(S(k),S(l)) = \frac{\rad_{\mathcal{S}}(S(k),S(l))}{\rad_{\mathcal{S}}^2 (S(k),S(l))}$$
where $\rad_{\mathcal{S}}(S(k),S(l))$ is the vector space of non-isomorphisms between $S(k)$ and $S(l).$  The space $\rad_{\mathcal{S}}(S(k),S(l))^2$ is the vector space of non-isomorphisms from
$S(k)$ to $S(l)$ admitting a non-trivial factorization
$$\rad_{\mathcal{S}}(S(k),S(l))^2 = \sum_{m} \rad(S(m),S(l)) \rad(S(k),S(m)).$$
Since all of the $S^T$ modules are generated by monomials, $\rad_{\mathcal{S}}^2$ will be generated by factorizations of the form $ \rad(S(m),S(l) \rad(S(k),S(m))$
without summing over possible intermediate $\gamma's.$

Determining the irreducible morphisms in the endomorphism algebra is the first step toward showing that the endomorphism algebra is equivalent to a superpotential algebra.  The endormophism algebra $A = \End_R(\mathcal{S})$ can be described by a quiver that has a vertex for each module $S(m)$ in the tilting complex.  The arrows between two vertices $S(k)$ and $S(l)$ in the quiver are chosen to form a basis of the space of irreducible morphisms $\Irr_{\mathcal{S}}(S(k),S(l)).$  Any morphism between two modules can be represented in the quiver as a path between the two modules.  In the next section we will give examples of the irreducible morphisms for two well known geometries.  We will also need to determine the relations in the endomorphism algebra to construct a superpotential that generates the same relations.  The relations in the endomorphism algebra come from two paths that represent the same morphism.  In section \ref{sec:morphisms} we will show the endomorphism algebra is isomorphic to the superpotential algebra of a periodic quiver. 

\section{Some Illustrative Examples}
\label{sec:examples}
\subsection{The Conifold}
In all of our examples we will rename the variables $(\alpha_1,\alpha_2,\beta_1,\beta_2) = (\alpha, \beta, \gamma, \delta)$ to simplify notation.
The conifold assigns the variables $(\alpha,\beta,\gamma,\delta)$ the weights $(1,1,-1,-1)$ respectively.  The two vertices correspond to the modules $S(0) = \C[\alpha \gamma,\alpha \delta,\beta \gamma,\beta \delta]$ and $S(1) = (\alpha,\beta) S(0).$  We use the letters $a,b,c,d$ to denote the morphism in $\End(\mathcal{S})$ induced by multiplication by $\alpha,\beta,\gamma$ or $\delta$ respectively.  The irreducible morphisms are shown in figure \ref{fig:conifold}.
\begin{figure}
\begin{center}
\includegraphics[width=12.5cm]{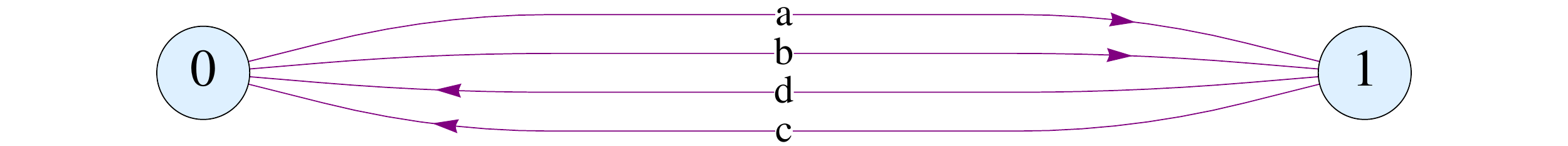}
\caption{Klebanov-Witten quiver for the conifold}
\label{fig:conifold}
\end{center}
\end{figure}
\subsection{$Y^{2,1}$}
\begin{figure}[h]
\begin{center}
\includegraphics[trim = 0mm 25mm 0mm 25mm, clip,width=9cm]{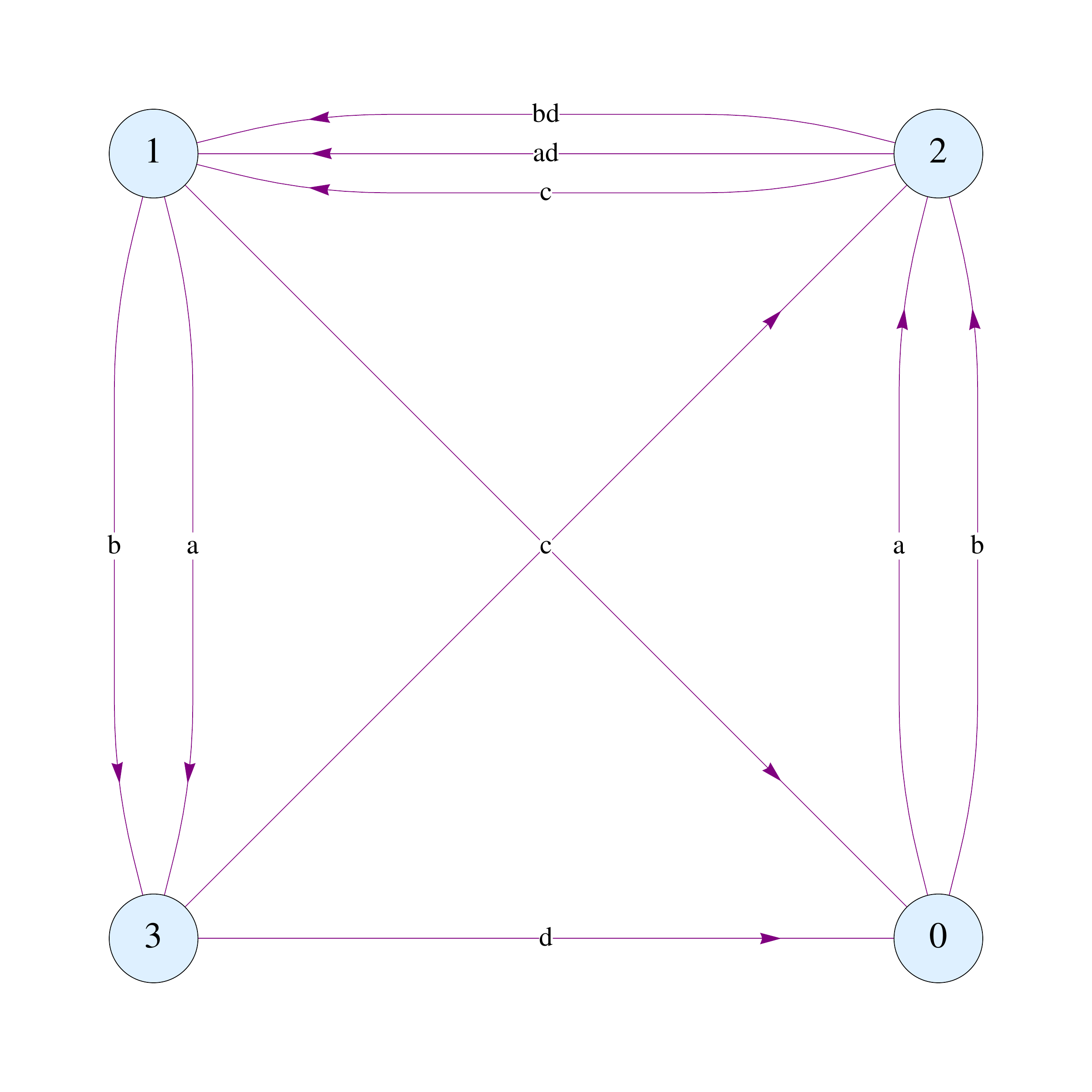}
\caption{Quiver for $Y^{2,1}$}
\label{fig:y21}
\end{center}
\end{figure}
In classical algebraic geometry, the cone over $Y^{2,1}$ is better known as the total space of the anti-canonical bundle over a particular K\"{a}hler base.  The K\"{a}hler base is the one point blow-up of $\P^2$ and is one of two isomorphism classes of degree 8 del Pezzo surfaces.  When the del Pezzo in the total space is blown-down to a point, the resulting geometry can be described by the GLSM that assigns weights $(2,2,-1,-3)$ to $(\alpha,\beta,\gamma,\delta)$.
The ring of invariants $S^{T}$ is generated by nine monomials
$$S^{T} = S(0) = \C[\beta^3 \delta^2, \alpha\beta^2 \delta^2,\alpha^2 \beta\delta^2,\alpha^3 \delta^2,\alpha\gamma^2,\beta\gamma^2,\beta^2 \gamma\delta,\alpha\beta\gamma\delta,\alpha^2\gamma\delta].$$
There are similar expressions for the modules corresponding to the other vertices.  However, we will see that their explicit form is unnecessary to determine the irreducible morphisms.
We depict all irreducible morphisms as arrows in the quiver shown in figure \ref{fig:y21}.  As an example, we now explain why the endomorphism $ad: S(2) \rightarrow S(1)$, induced from multiplication by $\alpha \delta$, is an irreducible morphism.  Multiplication by a single variable $\alpha$ or $\delta$ gives a map from $S(1)$ to $S(4)$ or $S(-1),$ but neither of these modules is a summand in the tilting complex $\mathcal{S} = S(0) \oplus S(1) \oplus S(2) \oplus S(3).$  Therefore $\alpha$ and $\delta$ do not induce morphisms in $\End(\mathcal{S}).$  Since it is impossible to decompose $ad$ into the composition of other morphisms, it is irreducible.  We can similarly  check that all of the other morphisms in figure \ref{fig:y21} are irreducible.  In the next section we will see that these are all of the irreducible morphisms as the consequence of a more general construction.
\section{Constructing the Periodic Quiver}
\label{sec:morphisms}
So far we have described an abstract procedure for determining a quiver with relations for a given $L^{a,b,c}$ singularity.  In this section we will show that the relations can be encoded by the superpotential of a periodic quiver.  Instead of describing the periodic quiver embedded in a torus, we will specify its lift to the universal cover, $\R^2$, of the torus.  We use a construction due to Speyer \cite{MR2317336}, which appeared in a seemingly unrelated context.

Let $\Lambda = \left\{(n,i,j) \in \Z^3 \vert n + i + j \equiv 0\pmod{2} \right\}$ be a face-centered cubic lattice.  The face-centered cubic lattice is tiled by octahedrons and tetrahedrons.  Over each point $(i,j) \in \Z^2$ there are precisely two values of $n$ satisfying
$$0 \le \frac{(2a - N)i + (2c - N)j + nN}{2} < N.$$ 
Of these two values of $n,$ precisely one triple $(n,i,j)$ belongs to the lattice $\Lambda.$  Call this value of $n$ the {\it height} $h(i,j)$ of the point $(i,j) \in \Z^2$ and define $\pi(i,j)$ to be the the unique value of
$\frac{(2a - N)i + (2c - N)j + nN}{2} $ in the range $[0,N).$  At each vertex of $\Z^2$ we associate the module $S(\pi(i,j))$ of semi-invariants.
\begin{figure}
\begin{center}
\includegraphics[trim = 0mm 0mm 0mm 0mm, clip,width=12cm]{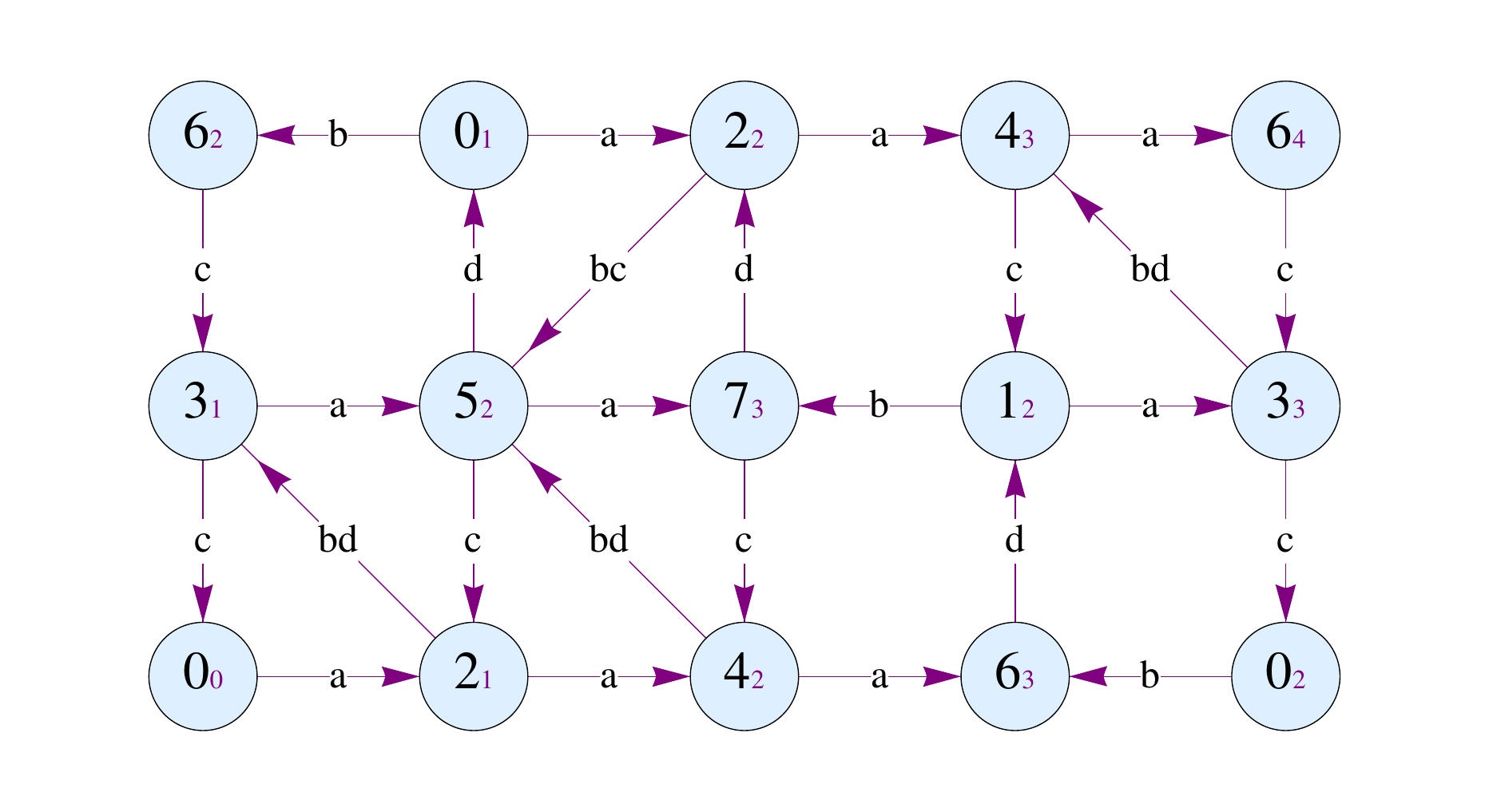}
\caption{Periodic quiver for $L^{2,6,3}$}
\label{fig:l263sq}
\end{center}
\end{figure}
An example of this construction is shown in figure \ref{fig:l263sq} for the $L^{2,6,3}$ singularity.  The lower left vertex has coordinate $(i,j) = (0,0)$, where $i$ is the horizontal coordinate and $j$ is the vertical coordinate. The large vertex labels are the functions $\pi(i,j)$ and the subscripts are the heights $h(i,j).$  We will next describe how the irreducible morphisms are represented by labeled arrows in the quiver. 
\begin{table}[htdp]
\begin{center}
\begin{tabular}{|c|c|}
\hline
 & $(\Delta n,\Delta i, \Delta j)$ \\
 \hline
 $a$ & (1,1,0) \\
 $b$ & (1,-1,0) \\
 $c$ & (-1,0,-1) \\
 $d$ & (-1,0,1) \\
\hline
\end{tabular}
\end{center}
\caption{Irreducible morphisms}
\label{tab:arrow}
\end{table}
Abusing notation, we use the name of a variable in the GLSM to also represent its weight.
Starting at a module $S(k)$, multiplication by a single variable $x = \alpha,\beta, \gamma,$ or $\delta$ induces a morphism in $\End(\mathcal{S})$ if $S(k + x)$ is a summand of the tilting module $\mathcal{S}$.  Multiplication by $x$ therefore induces a morphism if $0 \le k + x \le N-1.$
If $k = \pi(i,j)$, then a case-by-case analysis shows that $x$ is a morphism if and only if the height satisfies $h(i + \Delta i, j + \Delta j) = h(i,j) + \Delta n$, where the values of $(\Delta n, \Delta i, \Delta j)$ are listed in table \ref{tab:arrow}.

We now show all morphisms induced by the product of three or more variables are reducible.  Any morphism $x = x_1x_2 \dots x_n$ with $n \ge 3$ must have a factor $x_i = a,b$ in the horizontal direction and a factor $x_{j} = c,d$ in the vertical direction.  The height change between $h(i,j)$ and $h(i + \Delta i(x_i x_j), j + \Delta j(x_i, x_j))$ is either -2, 0, or 2.  A height change of -2 or 2 forces either of $x_i$ or $x_j$ to be a morphism.  If the height change is 0, then the composite $x_i x_j$ is a morphism.  In either case, $x$ can be factored into the product of morphisms and is therefore reducible. 
Irreducible morphisms of the form $x= ad$ from $\pi(i,j)$ to $\pi(i+1,j+1)$  arise precisely when $h(i + 1,j) - h(i, j+1) = 2.$  The possible differences in the height functions are shown in figure \ref{fig:heights}.  A similar analysis applies to the other three quadrants.  The result is that the types of vertices appearing in the quiver are precisely the same as those appearing in figure 5 of \cite{Franco:2005sm} and figure 4 of \cite{Benvenuti:2005ja}.
\begin{figure}[h]
\begin{center}
\includegraphics[width=12cm]{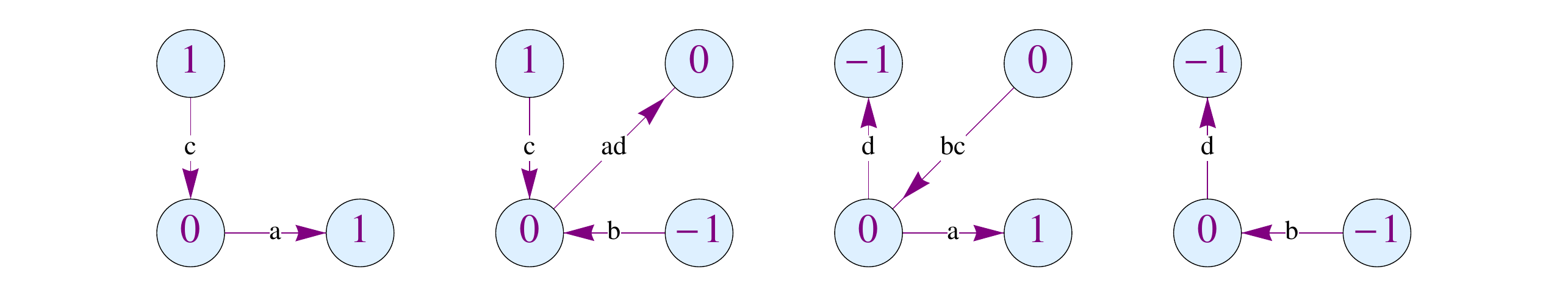}
\caption{Possible differences in height functions}
\label{fig:heights}
\end{center}
\end{figure}

We now show that the endomorphism algebra is isomorphic to the superpotential algebra with the same vertices and arrows, but with relations given by the superpotential
$$W = \sum_{F \in Q_{2}^{+}} w_F - \sum_{F \in Q_2^{-}} w_F.$$  A simple analysis with height functions shows that the product of morphisms around each face is always equal to $abcd.$  Therefore, the relations generated by the partial derivatives of the superpotential are automatically satisfied by the endomorphism algebra.  Conversely, if two morphisms from $S(k)$ to $S(l)$ are equal, then starting from any lift of $S(k)$ to the universal cover, the lifts of the two morphisms both end at the same vertex.  We can deform one path into the other by removing two adjacent faces at a time.  The superpotential relations imply that the two paths bounding the common edge are equal, so we obtain a sequence of equalities showing that the two paths are equal in the superpotential algebra $\C Q/(\partial W).$

\section{R-Charges and $U(1)$ Symmetries}
In this section we will give gauge theory interpretations of the vertex and edge labels used in the construction of the $L^{a,b,c}$ quivers.  If the $L^{a,b,c}$ quiver gauge theory flows to a superconformal field theory in the IR, then at the IR fixed point the NSVZ beta functions all vanish and the gauge groups all have equal rank.  A superconformal theory has global bosonic symmetry subgroup $SO(4,2) \times U(1)_R \subset SU(2,2 | 1)$ of the $\mathcal{N} = 1$ SCFT symmetry supergroup.  The $SO(4,2)$ factor is the ordinary conformal group of Minkowski space and the $U(1)_R$ symmetry rotates the two supersymmetry generators $Q_{\alpha}$ and $Q_{\dot{\alpha}}$ of the $\mathcal{N} = 1$ SUSY algebra into each other.  The $U(1)_R$ symmetry is a linear combination of the global $U(1)$ symmetries that come from the non-anomalous $U(1)$ gauge symmetries in the UV.  Our edge labeling of the $L^{a,b,c}$ quivers leads to a simple parametrization of the possible R-charge assignments for the bi-fundamental fields.

Determining the R-symmetry of a given SCFT is generally a difficult problem.  For a periodic quiver there are some simplifications because the superpotential terms are encoded in the faces of the periodic quiver.  To organize the geometric constraints, we follow the exposition in \cite{Mozgovoyy09} and start by defining a complex of abelian groups 
$$
\begin{CD}
 \Z_{Q_2} @> \partial_2 >> \Z_{Q_1} @> \partial_1>> \Z_{Q_0}  
\end{CD}
$$
where the differentials are $\partial_2(F)  = \sum_{a \in F}$ for $F \in Q_2$ a face and $\partial_1(a) = t(a) - h(a)$ for $a  \in Q$ an arrow.
Since $Q$ is a CW decomposition of the torus, the homology groups of this complex are equal to the homology groups of the torus.
Let $\Lambda$ be the $(|Q_0| + 1)$ dimensional lattice
$$\Lambda = \Z_{Q_1}/\langle \partial_2(F) - \partial_2(G) \;  \vert  \; F,G \in Q_2 \rangle $$
and denote the projection by
$\wt: \Z_{Q_1} \rightarrow \Lambda.$
The function $\wt$ is a weight function on paths that can be used to determine when two paths are equivalent under the F-term relations.
Let $\Lambda^{+} \subset \Lambda$ be the semigroup generated by the image of $\Z^{\ge 0}_{Q_1}$ under the map $\wt.$
Define the cone $P \subset \Lambda_{\mathbb{Q}}$ by
$$P = \left\{ \sum a_i \lambda_i \; \vert \; a_i \in \mathbb{Q}_{\ge 0}, \lambda_i \in \Lambda^{+} \text{ for all } i \right\}$$
The variety $\Spec \C[ P \cap \Lambda]$ has dimension $(|Q_0| + 2)$ and is the largest irreducible component of the master space \cite{Forcella:2008bb}.  This component is a  $(|Q_0| - 1)$ dimensional fibration
over the mesonic moduli space.  To understand the fibration, define
$B = \ker(\Z_{Q_0} \rightarrow \Z)$ where the map assigns 1 to each vertex.  The algebraic torus $T_B = \Hom_{\Z}(B, \C^{*})$ has a maximal compact subgroup $U(1)_B^{(Q_0 - 1)}$, which is the group of baryonic symmetries of the quiver gauge theory.  Define $M$ by the exact sequence
$$
\begin{CD}
0 @>>> M @> \iota >> \Lambda @> d >> B @>>> 0
\end{CD}
$$
where $d$ is the projection of $d_1$ under $\wt$ and $\iota$ is the inclusion map.  The lattice $M$ is three-dimensional and has a corresponding three-dimensional semigroup $M^{+} = M \cap \Lambda^{+}.$
Define the algebraic torus $T_{\Lambda}= \Hom_{\Z}(\Lambda, \C^{*})$, which is the combined mesonic and baryonic symmetries of the master space.
The $U(1)_M^3$ mesonic symmetries are the compact subgroup of the algebraic torus $T_{M} = T_{\Lambda}/T_{B}.$

The values of $(\Delta n, \Delta i, \Delta j)$ assigned to the irreducible morphisms are mesonic charges of the corresponding bifundamental fields.  Denote the mesonic $U(1)_M$ charges $(Q_N, Q_H, Q_V)$ where $H$ and $V$ stand for horizontal and vertical respectively.  The net change $(\Delta n, \Delta i, \Delta j)$ between two representatives of a vertex in the universal cover, $\Z^2$, of the periodic quiver is typically non-zero, so the $U(1)_M$ charges $(\Delta n, \Delta i, \Delta j)$ indeed correspond to mesonic symmetries.  In contrast, baryonic charges are defined to be charges that are equal for all representatives of a vertex.  Therefore the labels $\pi(i,j)$ are baryonic charges.  The baryonic charge of an edge is the difference of the two labels of the vertices the edge connects.  These differences are equal to the labels of the irreducible morphisms.

The chiral superfields are charged under a global $U(1)_R \subset U(1)_M^3$ symmetry.  The $U(1)_R$ charge assignments of the chiral superfields are given by a map
$$R: Q^{1} \rightarrow (0,1]$$
which must satisfy the following two geometric constraints \cite{Franco:2005rj}
\begin{equation}
\label{ra}
\sum_{a \in Face} R(a) = 2 \qquad \text{for all faces } F \in Q_2
\end{equation}
since the superpotential has $R$-charge 2 and
\begin{equation}
\label{rb}
\sum_{\text{edges } a \ni i} (1 - R(a)) = 2 \qquad \text{for all vertices } i \in Q_0
\end{equation}
since the NSVZ beta function is zero for each gauge group at the superconformal fixed point.  Such an R-charge assignment is called a {\it geometrically consistent R-charge} assignment by Mozgovoy \cite{Mozgovoyy09}.

Among all possible geometrically consistent R-charge assignments the one that maximizes the a-anomaly
\begin{equation}
a = \frac{9}{32} \left(Q_0 + \sum_{e \in Q_1} (R(e) - 1)^3  \right)
\end{equation}
is the physical R-charge at the superconformal fixed point \cite{Intriligator:2003jj}.

We now characterize all charge assignments satisfying equations \eqref{ra} and \eqref{rb} using perfect matchings \cite{Hanany:2005ve,Franco:2005rj}.  A {\it perfect matching} $M$ is a subset of the edges of a brane tiling such that every vertex belongs to precisely one edge of $M.$  Any charge assignment satisfying the first constraint \eqref{ra} can be written as a convex combination of charge assignments associated to perfect matchings of the brane tiling.  For every perfect matching $M$, define a function $\delta_M$ that takes the value 2 on the edges of $M$ and 0 on all other edges.  Any charge assignment can be represented by a matrix whose rows and columns  are indexed by the white and black vertices of the brane tiling and whose entries equal the R-charge of the edge connecting any pair of vertices.  King \cite{broomthesis, Forcella:2008bb} observed that the constraint \eqref{ra} implies that the row and column sums of the matrix are 2, which is by definition twice a doubly-stochastic matrix.  Any doubly stochastic matrix is a convex combination of permutation matrices by the Birkhoff-von Neumann theorem, but a permutation matrix is the R-charge assignment $\delta_M$ of a perfect matching $M$ so the result follows. 

To characterize perfect matchings, $M$, whose charge assignments satisfy \eqref{rb}, we introduce the notation of a simple perfect matching.  For any perfect matching $M$ we can construct a representation $\gamma_{M}$ with dimension vector $(1,1, \dots 1)$ of the superpotential algebra $\C Q/I$ by assigning each edge in the perfect matching 1 and all other edges 0.  A perfect matching $M$ is called {\it simple} if the corresponding quiver representation $\gamma_M$ is simple, that is, if $\gamma_{M}$ contains no non-trivial subrepresentations.  Given a perfect matching $M$ define $Q_M$ to be the quiver with the same vertices as $Q$ and only the arrows not contained in $M.$  Then $M$ is simple if and only if $Q_M$ is strongly connected \cite{ishiidimer}, i.e. there is an oriented path from any vertex to every other vertex.  We now state the relation between simple perfect matchings and solutions to $\eqref{rb}.$
\begin{lem}
The charge assignment $\delta_M$ of a simple perfect matching $M$ satisfies $\eqref{ra}$ and $\eqref{rb}.$  Conversely, every charge assignments satisfying $\eqref{ra}$ and $\eqref{rb}$ can be written as convex combinations of the functions $\delta_M$ for simple perfect matchings $M.$
\end{lem}
\begin{proof}
Since a simple perfect matching $M$ leaves all the faces of the brane tiling connected, each face in the brane tiling with $2n$ edges can have at most $n-1$ edges contained in $M.$  We rewrite this condition as 
$\sum_{e \in F} (1 - \delta_M(e)) \le 2.$  Following Gulotta's proof of Theorem 3.9 in \cite{Gulotta:2008ef}, we sum both sides of the inequality over all faces in the brane tiling
\begin{align}
\sum_{F  \in G_2} \sum_{e \in F} (1 - \delta_M(e)) & \le \sum_{F \in G_2} 2 \\
2 E - 2 \sum_{F  \in G_2} \# (e \in edges(F) \cap M) & \le 2F
\end{align}
but $ \sum_{F  \in G_2} \# (e \in edges(F) \cap M) = V$ since the number of edges in the brane tiling equals half the number of vertices.  Since the Euler characteristic of the torus is zero, $2E - 2V = 2F$, which forces  all of the inequalities to be equalities.  Therefore $\sum_{e \in F} (1 - \delta_M(e)) = 2$ for all faces in the brane tiling and $\delta_{M}$ satisfies \eqref{rb}.
\end{proof}
For the $L^{a,b,c}$ family of quivers we have constructed each superpotential term has R-charge $R(a) + R(b) + R(c) + R(d).$  Therefore the set of all edges containing a fixed variable $a$ $b$, $c$ or $d$ form a perfect matching $M.$  All four of these perfect matchings are simple so for any assignment of R-charges satisfying $R(a) + R(b) + R(c) + R(d) = 2,$ both \eqref{ra} and \eqref{rb} hold. 
\begin{table}[htdp]
\begin{center}
\begin{tabular}{|c|c|c|c|c|c|c|c|}
\hline
Field & Morphism & Multiplicity & $Q_B$ & $Q_N$ & $Q_H$ & $Q_V$ & R-charge \\
 \hline
 $\rightarrow$ & $a$ & $\beta$  & $+\alpha$ & 1 & 1 & 0 & $R(a)$ \\
 $\leftarrow$ & $b$  & $\alpha$ & $+\beta$ & 1 & -1 & 0 & $R(b)$ \\
 $\uparrow$ & $c$ & $\delta$ & $-\gamma$ & -1  & 0  & 1 & $R(c)$ \\
 $\downarrow$ & $d$ & $\gamma$ & $-\delta$ & -1 & 0 & -1 & $R(d)$ \\
 $\nwarrow$ & $bd$ & $\beta-\gamma$ & $\beta-\delta$ & 0 & -1 & 1& $R(b) + R(d)$  \\
 $\swarrow$ & $bc$ & $\beta-\delta$ & $\beta-\gamma$& 0 & -1 & -1 & $R(b) + R(c)$ \\
\hline
\end{tabular}
\end{center}
\caption{Charge assignments and multiplicities for basic fields \cite{Franco:2005sm, Benvenuti:2005ja, Butti:2005vn}.}
\label{tab:charge}
\end{table}%
The multiplicities of fields with common R-charge assignments together with their mesonic and baryonic charges are listed in table \ref{tab:charge} and are in perfect agreement with the results in \cite{Franco:2005sm, Benvenuti:2005ja, Butti:2005vn}.
\section{Seiberg Dualities and Cascades}
\label{sec:Seiberg}
Seiberg duality is a powerful tool in the study of $\mathcal{N} = 1$ gauge theories.  It is an equivalence between the IR dynamics of two different $\mathcal{N} = 1$ gauge theories with different UV descriptions \cite{Seiberg:1994pq}.  Taking the coupling of a gauge group to infinite coupling results in theory that has a dual description called its Seiberg dual.  The dual theory has as fundamental degrees of freedom dual quarks and mesons.  In a quiver gauge theory, we can view Seiberg duality as an operation on a single gauge group.  In terms of the quiver gauge theory, Seiberg duality at a vertex $k$ can be described by the following procedure \cite{Franco:2005rj}:
\begin{itemize}
\item Replace the gauge group at node $k$ with rank $N_C$ by one of rank $N'_C = N_F - N_C$ where $N_F$ is the number of flavors  and $N_C$ is the rank of the old gauge group.
\item Replace all quarks charged under the gauge group at node $k$ by ``dual-quarks'' transforming in the conjugate representation.  This corresponds to inverting all incoming and outgoing arrows at the node $k$. \item For every quark anti-quark pair $(Q_{\alpha} ,\widetilde{Q}^{\beta})$ represented by arrows $i \rightarrow k$ and $k \rightarrow j$, form a composite meson $M_{\alpha} ^{\beta} = Q_{\alpha} \widetilde{Q}^{\beta}.$  Represent this new meson by an arrow $i \rightarrow j$ in the Seiberg dual quiver.  The meson is now neutral under the gauge group of vertex $k$ and charged under the gauge groups corresponding to vertices $i$ and $j.$
\item Add cubic superpotential terms $\Delta W = Q'^{a} M_{\alpha}^{\beta} \widetilde{Q}'_{\beta}$ coupling the mesons and the dual quarks $Q'$.
\item Replace every composite operator $Q_{\alpha} \widetilde{Q}^{\beta}$ by the corresponding meson in the superpotential.  This step can cause cubic superpotential terms to turn into mass terms for some of the fields.  Integrate out these massive fields to find the superpotential of the Seiberg dual theory.
\end{itemize}
\begin{figure}
\begin{center}
\subfigure[Before Duality]{\includegraphics[width=6cm]{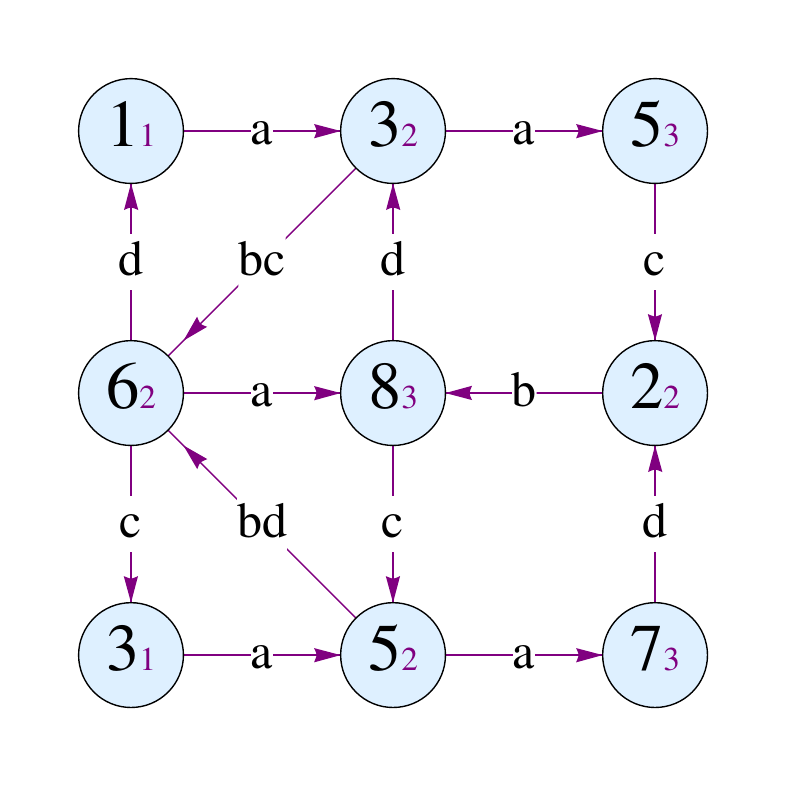}}
\qquad
\subfigure[After Duality]{\includegraphics[width=6cm]{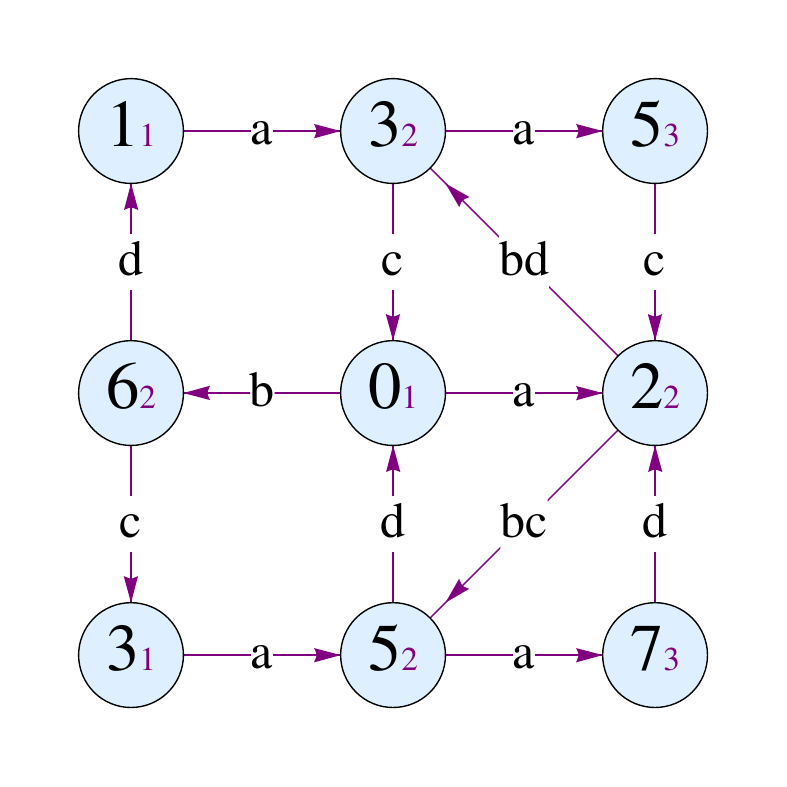}}
\caption{Seiberg duality applied to node 0}
\label{fig:SDtiling}
\end{center}
\end{figure}
Seiberg dualities that preserve the planar structure of the quiver are particularly simple to analyze.  The Seiberg dual of a periodic quiver is planar if the vertex being acted upon has two incoming and two outgoing arrows.  For the $L^{a,b,c}$ family, these vertices are precisely the local minima and maxima of the height function.  Seiberg duality on a local maximum acts by decreasing the height function by 2, which decreases the vertex's label $\pi(i,j)$ by $N$, as illustrated in figure \ref{fig:SDtiling}.

Performing a Seiberg duality on the vertex with $\pi(i,j) = N-1$ yields an identical quiver after cyclicly relabeling the nodes $N-1 \rightarrow N-2, N-2 \rightarrow N-3, \dots 0 \rightarrow -1.$  Seiberg duality changes the window of allowed values of $\pi(i,j)$ from $[0, N-1]$ to $[-1, N-1].$  When the ranks of all the gauge groups are equal, the gauge group couplings do not run and are parameters of the theory.  Taking one of the couplings to infinity results in a strongly coupled theory that has the same IR physics as its weakly coupled Seiberg dual. 

A collection of fractional branes must satisfy the condition for vanishing gauge anomalies \eqref{eqn:rank}.  When the ranks of the gauge groups are unequal, the NSVZ beta functions no longer vanish and the gauge group couplings can run.  When one of the gauge group couplings becomes infinitely strong, Seiberg duality replaces the gauge group by one of lower rank, which is weakly coupled compared to the other gauge groups.  Any rank assignment that is a linear combination of the non-anomalous $U(1)$ symmetries with total weight zero satisfies equation \eqref{eqn:rank} \cite{Herzog:2003dj}.  Therefore $n_v = N_C + k \pi_v$ is a valid rank assignment, where $\pi_v$ is the common label $\pi(i,j)$ for the vertex $v$'s representatives in the universal cover of the quiver.  Seiberg duality acting on the gauge groups with the largest rank decreases its rank by $kN$ and yields an identical quiver after relabeling the nodes, but with all the ranks of the gauge groups decreased by $k$.  This is the first step in proving the existence of the duality cascade whose existence was conjectured in \cite{Butti:2005sw,Brini:2006ej,Evslin:2007au}.  Checking that this procedure results in a duality cascade would require finding initial conditions for the gauge couplings that allow them to become strongly coupled in the correct order for the duality cascade to proceed.  We expect that the cascade would be qualitatively similar to the $Y^{p,q}$ cascades \cite{Herzog:2004tr, Benvenuti:2004wx, Brini:2006ej}.

Berenstein and Douglas \cite{Berenstein:2002fi} interpret Seiberg dualities in terms of tilting equivalences of derived categories using Rickard's theory of derived Morita equivalence \cite{MR1002456}.  Parts of Rickard's theory are introduced in appendix \ref{app-Morita}.  Using a Koszul exact sequence, it should be possible to verify that the Seiberg duality operation described here corresponds to a tilting equivalence.  

\section{Conclusion}
\label{sec:conclusion}
We have derived a very simple description of the $Y^{p,q}$ and $L^{a,b,c}$ quiver gauge theories and have found that they all possess periodic Seiberg dualities.   Our description of the periodic quiver was based on a construction of Speyer \cite{MR2317336} used to give a combinatorial proof of the integrality of the terms in the Gale-Robinson recurrence
$$g(n+N) = \frac{g(n+a) g(n+b) + g(n+c)g(n+d)}{g(n)}.$$
Speyer showed the terms of this sequence count perfect matchings on a subset of the the $L^{a,b,c}$ brane tiling.  For example, the recurrence for the conifold is $g(n+2)g(n) = 2 g(n+1)^2$, which counts the number of domino tilting of an order $n$ Aztec diamond \cite{MR1226347}.  Exploring the relationship between the domino tilings of Speyer and the perfect matchings arising in the study of non-commutative Donaldson-Thomas invariants \cite{MR2403807,Mozgovoy:2008fd,Ooguri:2008yb} is an exciting direction for future study.

Superpotential algebras consisting of vertical, horizontal, and diagonal edges have been previously studied in \cite{Beil:2008ck} in connection with the $Y^{p,q}$ family of quivers.  Perhaps the height function introduced in this paper could help shed new light on these non-commutative algebras.

Our graphical presentation of the quiver gauge theory may help show that these quivers are equivalent to the quivers obtained from algorithms using coamoebas or alga \cite{Feng:2005gw}.  Hanany and Vegh developed a method \cite{Hanany:2005ss} to determine the brane tiling for any local toric Calabi-Yau threefold using the normal vectors to the boundary of its toric diagram.  Their method is a conjectural way of extracting a gauge theory on a Calabi-Yau from the coamoeba of its mirror Calabi-Yau.  Developing the relationship between the quivers obtained from the theory of tilting and those obtained from coamoebas will hopefully lead to a deeper understanding of the physics of D3-branes at a toric Calabi-Yau singularity. 

\section{Acknowledgments}
I would like to thank Paul Aspinwall, David Berenstein, Aaron Bergman, Sebastian Franco, and David Speyer for discussions that helped improve this paper.  I would like to thank MSRI for inviting me to participate in the Tropical Structures in Geometry and Physics conference where part of this work was completed.  Finally, I would like to thank David Morrison for his detailed reading of this manuscript and numerous suggestions.  This research was supported in part by the National Science Foundation under Grant No. DMS-0606578.
\appendix
\section{Tilting Equivalences}
\label{app-Morita}
Rickard's theory of derived Morita equivalence helps us understand when two different rings $R$ and $S$ have the same derived categories of modules.  We will apply this technology to rings that are path algebras of quivers.  A nice introduction to the theory of tilting is the book \cite{MR1649837}.  An introduction to tilting in the context of Seiberg duality is given by Vit\'{o}ria \cite{MR2488553} and has been recently generalized by Keller and Yang \cite{keller08, kelleryang}.

Let $R$ be and ring and denote by $P(R)$ the category of finitely generated right projective modules over $R.$
\begin{defn}
A {\it tilting complex} over a ring $R$ is an object $T$ of the bounded homotopy category $K^b(P(R))$ such that
\begin{itemize}
\item $\Hom_{K^b(P(R))}(T,T[i]) = 0$ for all $i \neq 0$
\item $T$ generates $K^b(P(R))$ as a triangulated category.
\end{itemize}
\end{defn}
\begin{thm}[Rickard]
Let $R$ and $S$ be two rings.  Then $D^b(R)$ is derived equivalent to $D^b(S)$ if and only if there exists a tilting complex $T$ over $R$ such that
$$S \cong \End_{K^b(R)}(T)^{op}$$
\end{thm}
Vitoria shows that a tilting complex that generates Seiberg duality at node $k$ is
$$T = \bigoplus_{i=1}^n T_i$$
where 
$T_i = 0 \rightarrow P_i \rightarrow 0$
is the stalk complex for $i \neq k$ and for $i = k,$
$$T_k = 0 \rightarrow \bigoplus_{j \rightarrow k} P_j \rightarrow P_k \rightarrow 0.$$
\bibliography{BraneTiling}
\end{document}